\documentclass[journal]{IEEEtran}

\newtheorem{prop}{Proposition}

\newtheorem{cor}{Corollary}

\newtheorem{lm}{Lemma}

\newtheorem{thm}{Theorem}

\newcommand{\bal}{\begin{align}}
\newcommand{\eal}{\end{align}}
\newcommand{\be}{\begin{eqnarray}}
\newcommand{\ee}{\end{eqnarray}}
\newcommand{\benn}{\begin{eqnarray*}}
\newcommand{\eenn}{\end{eqnarray*}}
\def\IR{\rm I \kern-0.20em R}
\newcommand{\utwi}[1]{\mbox{\boldmath $ #1$}}

\newcommand{\bthm}{\begin{thm}}
\newcommand{\ethm}{\end{thm}}

\newcommand{\bcor}{\begin{cor}}
\newcommand{\ecor}{\end{cor}}
\newcommand{\bprop}{\begin{prop}}
\newcommand{\eprop}{\end{prop}}
\newcommand{\blm}{\begin{lm}}
\newcommand{\elm}{\end{lm}}
\newcommand{\beq}{\begin{equation}}
\newcommand{\eeq}{\end{equation}}
\newcommand{\ber}{\begin{eqnarray}}
\newcommand{\eer}{\end{eqnarray}}

\newcommand{\bproof}{\begin{proof}}
\newcommand{\eproof}{\end{proof}}

\newcommand{\diag}{\mathop{\mbox{\rm diag}}}

\newcommand{\bit}{\begin{itemize}}
\newcommand{\eit}{\end{itemize}}
\newcommand{\ben}{\begin{enumerate}}
\newcommand{\een}{\end{enumerate}}
\newcommand{\bdesc}{\begin{description}}
\newcommand{\edesc}{\end{description}}
\newcommand{\beqarrn}{\begin{eqnarray*}}
\newcommand{\eeqarrn}{\end{eqnarray*}}
\newcommand{\bproofof}{\begin{proofof}}
\newcommand{\eproofof}{\end{proofof}}
\newenvironment{rem}{\begin{trivlist}\item[]{\bf
Remark:}\hspace{4mm}}{\end{trivlist}}
\newcommand{\brem}{\begin{rem}}
\newcommand{\erem}{\end{rem}}
\newenvironment{rems}{\begin{trivlist}\item[]{\bf
Remarks}\begin{itemize}}{\end{itemize}\end{trivlist}}
\newcommand{\brems}{\begin{rems}}
\newcommand{\erems}{\end{rems}}
\newtheorem{fact}{Fact}
\newcommand{\bfact}{\begin{fact}}
\newcommand{\efact}{\end{fact}}
\newtheorem{examp}{Example}
\newcommand{\bexamp}{\begin{examp}\rm}
\newcommand{\eexamp}{\end{examp}}
\newtheorem{defn}{Definition}
\newcommand{\bdefn}{\begin{defn}\rm}
\newcommand{\edefn}{\end{defn}}

\newtheorem{alg}{Algorithm}
\newcommand{\balg}{\begin{alg}}
\newcommand{\ealg}{\end{alg}}

\newtheorem{prob}{Problem}
\newcommand{\bprob}{\begin{prob}}
\newcommand{\eprob}{\end{prob}}

\newcommand{\bvtm}{\begin{verbatim}}
\newcommand{\bfig}{\begin{figure}}
\newcommand{\efig}{\end{figure}}
\newcommand{\bcen}{\begin{center}}
\newcommand{\ecen}{\end{center}}

\long\def\comment#1{}

\def \n2{{N_0 \over 2}}

\def \h5{\hspace{0.5in}}

\newcommand{\ba}{{\utwi{a}}}
\newcommand{\bb}{{\utwi{b}}}
\newcommand{\bc}{{\utwi{c}}}

\newcommand{\bff}{{\utwi{f}}}

\newcommand{\br}{{\utwi{r}}}

\newcommand{\bu}{{\utwi{u}}}
\newcommand{\bv}{{\utwi{v}}}
\newcommand{\bw}{{\utwi{w}}}
\newcommand{\bx}{{\utwi{x}}}
\newcommand{\by}{{\utwi{y}}}

\newcommand{\bA}{{\utwi{A}}}
\newcommand{\bB}{{\utwi{B}}}

\newcommand{\bE}{{\utwi{E}}}

\newcommand{\bG}{{\utwi{G}}}
\newcommand{\bH}{{\utwi{H}}}
\newcommand{\bI}{{\utwi{I}}}

\newcommand{\bM}{{\utwi{M}}}
\newcommand{\bN}{{\utwi{N}}}

\newcommand{\bP}{{\utwi{P}}}
\newcommand{\bQ}{{\utwi{Q}}}
\newcommand{\bR}{{\utwi{R}}}
\newcommand{\bS}{{\utwi{S}}}

\newcommand{\bX}{{\utwi{X}}}

\usepackage{tikz}
\usetikzlibrary{arrows}
\tikzstyle{block}=[draw opacity=0.7,line width=1.4cm]

\newtheorem{lemma}{Lemma}

\newtheorem{remark}{Remark}
\newtheorem{proposition}{Proposition}

\usepackage{latexsym,amssymb,amsmath}
\usepackage{graphicx,mathtools}
\usepackage{cite,array,booktabs}

\title{Wireless MIMO Switching with Zero-forcing Relaying and Network-coded Relaying}

\author{Fanggang~Wang, \emph{Member, IEEE}, Soung Chang~Liew, \emph{Fellow, IEEE}, and Dongning~Guo, \emph{Senior Member, IEEE}
\thanks{Manuscript received August 13, 2011; revised December 21, 2011; accepted May 5, 2012. This work was partially supported by grants from the Univ. Grants Committee of the Hong Kong, China (AoE/E-02/08; 414911); the State Key Lab of Rail Traffic Control and Safety (RCS2011ZT011); the Fundamental Research Funds for
the Central Universities (2011JBM203); Program for Changjiang
Scholars and Innovative Research Team in Univ. (IRT0949); the Joint
Funds of State Key Program of NSFC (60830001).

F. Wang is with the State Key Lab of Rail Traffic Control and Safety, School of Electronic and Information Engineering, Beijing Jiaotong University, Beijing, China, and Institute of Network Coding, The Chinese University of Hong Kong, HK SAR, China (e-mail: fgwang@inc.cuhk.edu.hk).

S. Liew is with the Department of Information Engineering, the Chinese University of Hong Kong, HK SAR, China (e-mail:   soung@ie.cuhk.edu.hk).

D. Guo is with the Department of Electrical Engineering \& Computer Science,
Northwestern University, Evanston, IL, USA (e-mail: dGuo@northwestern.edu).
}}

\begin{document}
\IEEEoverridecommandlockouts

\maketitle
\thispagestyle{empty}

%
%

\begin{abstract} \label{abs}
A wireless relay with multiple antennas is called a multiple-input-multiple-output (MIMO) switch if it maps its input links to its output links using ``precode-and-forward.'' Namely, the MIMO switch precodes the received signal vector in the uplink using some matrix for transmission in the downlink. This paper studies the scenario of $K$ stations and a MIMO switch, which has full channel state information. The precoder at the MIMO switch is either a zero-forcing matrix or a network-coded matrix. With the zero-forcing precoder, each destination station receives only its desired signal with enhanced noise but no interference. With the network-coded precoder, each station receives not only its desired signal and noise, but possibly also self-interference, which can be perfectly canceled.
Precoder design for optimizing the received signal-to-noise ratios at the destinations is investigated. For zero-forcing relaying, the problem is solved in closed form in the two-user case, whereas in the case of more users, efficient algorithms are proposed and shown to be close to what can be achieved by extensive random search. For network-coded relaying, we present efficient iterative algorithms that can boost the throughput further.
\end{abstract}


%

\section{Introduction}
Relays in wireless networks can extend coverage as well as improve
energy efficiency \cite{den09}. In this paper, we study a setup in
which multiple single-antenna stations communicate with each other via a
multi-antenna relay. In each uplink slot, the stations simultaneously transmit, then in the subsequent downlink slot, the relay precodes its received signal by a certain
matrix before broadcasting to the stations. In the absence of noise,
the multiple-input multiple-output (MIMO) system between the
transmitters and the receivers can be viewed as a product of the
downlink channel matrix, the precoder and the uplink channel matrix.
In this work, we  design a zero-forcing
precoder so that the product channel is a desired permutation matrix, which forms a one-to-one mapping (or links) from the transmitters to the receivers.  Hence the technique is called MIMO switching.  We also study a generalization referred to as network-coded precoding where the off-diagonal elements of the channel matrix form a permutation, and where the diagonal elements can be nonzero.  Such nonzero diagonal elements cause self-interference, which can be fully canceled assuming the interference gains are available at the receivers.  We study how to design the precoder to maximize the signal-to-noise ratios (SNRs) of the links.


To the best of our knowledge, this work is the first to treat unicast non-pairwise switching patterns. Prior work that investigate data exchange via a relay includes \cite{den09, guo08, dguo12jsac, Cui08, Gao09, moh09}. References \cite{guo08, Cui08, Gao09} investigate the case of ``full data exchange,'' in which all stations want to broadcast their data to all the other stations. A slotted system with a single-antenna relay is considered in \cite{guo08} and the maximum throughput region is evaluated. Data transmissions in \cite{Cui08, Gao09} can be summarized as follows: In the first slot, all stations transmit to the relay simultaneously; subsequent slots are devoted to downlink transmissions; in each downlink slot, the relay multiplies the signal received in the first time slot by a different precoder, such that at the end of all downlink slots, all stations receive the broadcast data from all the other stations. By contrast, the framework investigated in this paper focuses on the unicast case, in which station $i$ transmits to another station $j$ only. (Station $j$ may transmit to a different station than $i$.) Any general transmission pattern (unicast, multicast, broadcast, or a mixture of them) among the stations can be realized by scheduling a set of different unicast transmissions, as has been pointed out by the authors in preliminary work \cite{allerton11}. A single-antenna relay with different forwarding strategies is considered in \cite{den09}, which studies both full data exchange and ``pairwise data exchange,'' in which stations form pairs to exchange data with each other only. It is a special case of unicast switching studied here. Reference \cite{moh09} studies pairwise data exchange only, where the relay adopts the decode-and-forward strategy. The diversity-multiplexing tradeoff under reciprocal and non-reciprocal channels is also analyzed.

In this paper, we consider both pairwise and non-pairwise switching, in which a multi-antenna relay works in precode-and-forward manner. We first study switching traffic among the stations using a zero-forcing MIMO relay, where each destination receives the desired signal with enhanced noise. We then study a more general network-coded relaying, which exploits physical-layer network coding for performance improvement \cite{Ahl00,Li03,Yeung06,Zhang06physical-layernetwork}. For fairness, we study how to design the precoder to maximize the minimum received SNR among all stations, which is referred to as the {\em maxmin problem}. Since the maxmin problem is NP-hard, we use a semidefinite relaxation technique to compute an approximate solution. The problem further simplifies if the SNRs at all destinations are required to be identical, and we call it the {\em equal-SNR problem}. We derive conditions under which the maxmin and equal-SNR problems are equivalent. By evaluating the throughput performances of the two problems, we show the gap between them is small especially in the high SNR regime. That is, our numerical results suggest that we can use the equal-SNR problem to approximate the (NP-hard) maxmin problem. Furthermore, we show that network-coded relaying can noticeably improve the throughput performance over zero-forcing relaying.

The remainder of the paper is organized as follows: Section II introduces the scheme of wireless MIMO switching. In Section III (resp.\ Section IV), the maxmin (resp.\ equal-SNR) problem is investigated for both zero-forcing and network-coded relaying. Section V presents the simulation results. Section VI concludes this paper.

\tikzstyle{place}=[circle,draw=blue!50,fill=blue!20,thick,
inner sep=0pt,minimum size=6mm]
\tikzstyle{transition}=[rectangle,draw=black!100,fill=none,semithick,
inner sep=0pt,minimum size=6mm]
\tikzstyle{mimoswitch}=[rectangle,draw=black!100,fill=none,semithick,
inner sep=0pt,minimum width=12mm, minimum height=26mm]
\tikzstyle{relay}=[rectangle,draw=black!100,fill=none,semithick,
inner sep=0pt,minimum width=25mm, minimum height=10mm]

\def\antenna{%
    -- +(0mm,4.0mm) -- +(2.625mm,7.5mm) -- +(-2.625mm,7.5mm) -- +(0mm,4.0mm)
}

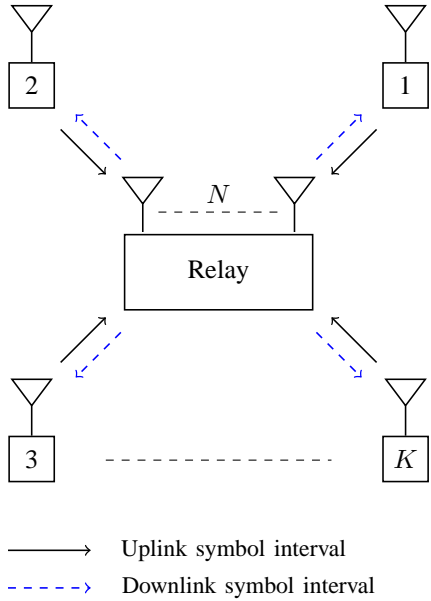
\begin{figure}
\centering
\begin{tikzpicture}[node distance=10em]
\node (relay) [relay] {Relay};
\node (1) [above right of=relay] [transition] {1};
\node (2) [above left of=relay] [transition] {2};
\node (3) [below left of=relay] [transition] {3};
\node (N) [below right of=relay] [transition] {$K$};

\draw[color=black,semithick] (1.north) \antenna;
\draw[color=black,semithick] (2.north) \antenna;
\draw[color=black,semithick] (3.north) \antenna;
\draw[color=black,semithick] (N.north) \antenna;
\path (relay.north west) to node (a) [pos=.2,inner sep=0] {} (relay.north) to node (b) [pos=.8,inner sep=0] {} (relay.north east);
\draw[color=black,semithick] (a) \antenna  (b) \antenna;

\draw [dashed] (-.8,.8) to node [above]  {$N$} (.8,.8);
\draw [dashed] (-1.5,-2.5) to node [auto] {} (1.5,-2.5);
\draw [color=blue, dashed, semithick, ->] (1.3,1.5) -- (1.9,2.1); \draw [semithick, ->] (2.1,1.9) -- (1.5,1.3);
\draw [color=blue, dashed, semithick, ->] (-1.3,1.5) -- (-1.9,2.1); \draw [semithick, ->] (-2.1,1.9) -- (-1.5,1.3);
\draw [color=blue, dashed, semithick, ->] (1.3,-0.8) -- (1.9,-1.4); \draw [semithick, ->] (2.1,-1.2) -- (1.5,-0.6);
\draw [color=blue, dashed, semithick, ->] (-1.3,-0.8) -- (-1.9,-1.4); \draw [semithick, ->] (-2.1,-1.2) -- (-1.5,-0.6);

\draw [color=blue, dashed, semithick, ->] (-2.8,-4.2) -- (-1.7,-4.2); \draw [semithick, ->] (-2.8,-3.7) -- (-1.7,-3.7);
\node at (0.2,-3.7) {\small{Uplink symbol interval}};
\node at (0.4,-4.2) {\small{Downlink symbol interval}};
\end{tikzpicture}
\caption{Wireless MIMO switching.}
\label{diagram}
\end{figure}

\section{System Description} \label{sec.SystemDes}

Consider $K$ stations, numbered $1, \dots, K$, each with one antenna, as shown in Fig.~\ref{diagram}. There is no direct link between any two stations and the stations communicate via a relay with $N$ antennas.
The precode-and-forward scheme applies under the condition of $K\leq N$, where the relay has enough degrees of freedom to switch all data streams at the same time.
We assume $K=N$ throughout for simplicity. In the case of $K < N$, all the matrix inverses in the paper shall be replaced by Moore-Penrose pseudo inverses \cite{matrix}. Each transmission consists of one uplink symbol interval and one downlink symbol interval of equal duration. In particular, the two symbol intervals are two slots in a time-division system. The uplink symbol interval is for simultaneous uplink transmissions from the stations to the relay; the downlink symbol interval is for downlink transmissions from the relay to the stations. Each round of uplink and downlink transmission realizes a switching permutation, as shall be described shortly.

Consider one transmission. Let $\bx = [x_1, \cdots, x_N]^T$ be the vector representing the signals transmitted by the stations. Let $\by=[y_1, \cdots, y_N]^T$ be the received signals at the relay, and $\bu=[u_1, \cdots, u_N]^T$  be the noise vector with independent identically distributed (i.i.d.) noise samples following circularly-symmetric complex Gaussian (CSCG) distribution, i.e., $u_n \sim \mathcal{N}_c(0, \gamma^2)$. Then
\begin{equation} \label{formula_uplink}
\by=\bH \bx +\bu,
\end{equation}
where $\bH$  is the uplink channel gain matrix.
The relay multiplies $\by$  by a precoding matrix $\bG$ before relaying the signals. In this paper, we assume that the uplink channel and downlink channel are reciprocal, so that the downlink channel is $\bH^T$. Thus, the received signals at the stations in vector form are
\bal \label{formula_receive}
\br =& \bH^T \bG\by + \bw\\
 =& \bH^T \bG\bH \bx + \bH^T \bG\bu  + \bw,
\end{align}
where $\bw$ is the noise vector at the receiver, with the i.i.d. noise samples following CSCG distribution, i.e., $w_n \sim \mathcal{N}_c(0, \sigma^2)$.

In the following, we describe two precoding schemes.

\subsection{Zero-forcing Relaying}
We refer to an $N\times N$ matrix $\bP$ that has one and only one nonzero element on each row and each column, which is equal to $1$, as a {\em permutation} matrix. Evidently, $\bP\bx$ is a column vector consisting of the same elements as $\bx$ but permuted in a certain order depending on $\bP$.  For example, if
\[
\bP =
\begin{bmatrix}
0 & 0 & 1\\
1 & 0 & 0\\
0 & 1 & 0\\
\end{bmatrix},
\]
then $\bP[x_1, x_2, x_3]^T = [x_3, x_1, x_2]^T$. In the case where all diagonal elements of $\bP$ are zero it is also called a {\em derangement}.

Suppose that the purpose of $\bG$ is to realize a particular permutation represented by the permutation matrix $\bP$, and to amplify the signals coming from the stations. That is,
\begin{equation} \label{4}
\bH^T\bG\bH=\bA\bP,
\end{equation}
where $\bA=\diag\{a_1,\cdots,a_N\}$ is an ``amplification'' diagonal matrix. Each diagonal element is regarded as the gain of a link. Accordingly, the precoder can be calculated as
\be \label{10}
\bG=\bH^{-T}\bA\bP\bH^{-1}.
\ee
Let the receivers compensate for the amplification to yield received signals expressed collectively as:
\be
\hat{\br} = \bA^{-1} \br = \bP \bx + \bv,
\ee
where the {\em post-processing} noise is expressed as
\begin{align}
\bv = \bP \bH^{-1}\bu + \bA^{-1}\bw.
\end{align}
Let us define
\be \label{dd}
\bQ \triangleq \bI + \gamma^2\bP\bH^{-1}\bH^{-H}\bP^T.
\ee
The covariance of the post-processing noise $\bv$ is written as
\bal  \label{8}
\mathbb{E}\{\bv\bv^H\}=&\gamma^2 \bP\bH^{-1}\bH^{-H}\bP^T + \sigma^2 \bA^{-1}\bA^{-H}\\
=&\bQ -\bI + \sigma^2 \bA^{-1}\bA^{-H}.
\end{align}

Suppose all uplink transmissions are independent and use unit average power, i.e., $\mathbb{E}\{x_i^2\} = 1,\ i=1,\cdots,N$.
The problem is to design the precoder $\bG$ to minimize the variance of the post-processing noise subject to a power constraint for the relay:
\be
\mathbb{E}\{\by^H\bG^H\bG\by\} \le p.
\ee

For notational convenience, let the entries of an $N\times N$ matrix $\bS$ be given by
\begin{align} \label{def_S}
s_{ij} \triangleq \bQ_{ji} [(\bH^*)^{-1}\bH^{-T}]_{ij}.
\end{align}
From \eqref{formula_uplink}, the relay's transmit power can be evaluated as
\begin{align}  \label{9}
\mathbb{E}[&\bx^H \bH^H  \bG^H \bG \bH \bx + \bu^H \bG^H \bG \bu ]\notag\\
=&\text{Tr}[\bG\bH\bH^H\bG^H+\gamma^2 \bG\bG^H]\\
=&\text{Tr}[\bH^{-T} \bA (\bI + \gamma^2\bP\bH^{-1}\bH^{-H}\bP^T)\bA^H (\bH^*)^{-1} ]\\
=&\text{Tr}[\bA \bQ \bA^H(\bH^*)^{-1}\bH^{-T}]\\
=&\ba^H\bS\ba,
\end{align}
where we have used \eqref{10} and \eqref{dd}, and $\ba=[a_1,\cdots,a_N]^T$ is the gain vector with the diagonal elements of $\bA$. The power constraint on the relay is thus expressed as
\begin{align}
\ba^H\bS\ba \leq p.
\end{align}

\tikzstyle{transition}=[rectangle,draw=black!100,fill=none,semithick,
inner sep=0pt,minimum size=6mm]
\tikzstyle{mimoswitch}=[rectangle,draw=black!100,fill=none,semithick,
inner sep=0pt,minimum width=12mm, minimum height=26mm]

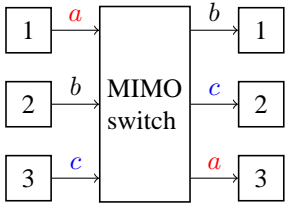
\begin{figure}
\centering
\begin{tikzpicture}[x=2.2em,y=2.8em]
\node (1s1) at (-1,2) [transition] {1};
\node (1s2) at (-1,1) [transition] {2};
\node (1s3) at (-1,0) [transition] {3};

\node (2s3) at (3,2) [transition] {1};
\node (2s1) at (3,1) [transition] {2};
\node (2s2) at (3,0) [transition] {3};

\node (m1) at (1,1) [mimoswitch] {\parbox{10mm}{MIMO \\ switch}};

\path[->] (1s1) edge node [above]{$\color{red}a$} (m1.west|-1s1);
\path[->] (m1.east|-2s3) edge node [above]{$b$} (2s3);
\path[->] (1s2) edge node [above]{$b$} (m1.west|-1s2);
\path[->] (m1.east|-2s1) edge node [above]{$\color{blue}c$} (2s1);
\path[->] (1s3) edge node [above]{$\color{blue}c$} (m1.west|-1s3);
\path[->] (m1.east|-2s2) edge node [above]{$\color{red}a$} (2s2);

\end{tikzpicture}
\caption{A traffic demand among three stations.}
\label{3node_realization}
\end{figure}

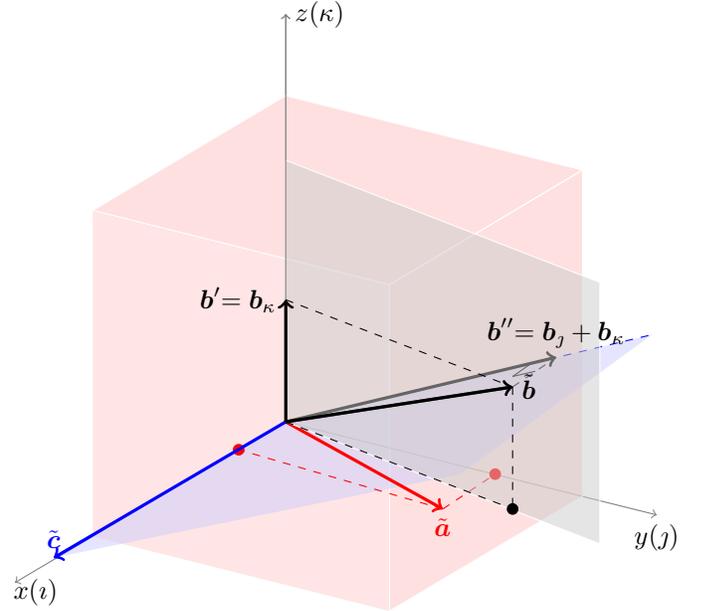
\begin{figure}
\centering
\begin{tikzpicture}[line join=round,x=6.6em,y=6.6em]
\draw[arrows=-,color=gray](0,0)--(-.553,-.329);
\draw[arrows=-,color=gray](0,0)--(0,.936);
\draw[arrows=-,color=gray](0,0)--(1.701,-.428);
\draw[arrows=->,color=gray](0,.936)--(0,2.339);
\draw[arrows=-,color=gray](-.553,-.329)--(-1.106,-.658);
\draw[arrows=->,color=gray](1.701,-.428)--(2.126,-.534);
\draw[arrows=->,color=gray](-1.106,-.658)--(-1.552,-.922);
\filldraw[draw=none,fill=red!20,fill opacity=0.5](-1.106,-.658)--(0,0)--(1.701,-.428)--(.595,-1.085)--cycle;
\filldraw[draw=none,fill=red!20,fill opacity=0.5](-1.106,-.658)--(0,0)--(0,1.871)--(-1.106,1.213)--cycle;
\filldraw[draw=none,fill=red!20,fill opacity=0.5](0,1.871)--(0,0)--(1.701,-.428)--(1.701,1.444)--cycle;
\filldraw[draw=white,fill=red!20,fill opacity=0.2](.595,.786)--(1.701,1.444)--(0,1.871)--(-1.106,1.213)--cycle;
\filldraw[draw=white,fill=red!20,fill opacity=0.2](.595,.786)--(1.701,1.444)--(1.701,-.428)--(.595,-1.085)--cycle;

\filldraw[draw=none,fill=blue!20,fill opacity=0.5](0,0)--(-1.33,-.78)--(1.0,-.3)--(2.1,.5)--cycle;

\draw[draw=red,very thick,->](0,0)--(.9,-.5);
\draw[draw=red,dashed](.9,-.5)--(-.27,-.16);
\draw[draw=red,dashed](.9,-.5)--(1.2,-.3);
\filldraw[color=red](-.27,-.16) circle (2pt);
\filldraw[color=red](1.2,-.3) circle (2pt);

\draw[draw=blue,very thick,->](0,0)--(-1.33,-.78);

\draw[draw=black, very thick,->](0,0)--(1.3,.2);
\draw[draw=blue,dashed](0,0)--(2.1,.5);
\draw[draw=black,dashed](1.3,.2)--(1.55,.37);
\draw[draw=black,very thick, ->](0,0)--(1.55,.37);
\draw[draw=black](1.3,.26)--(1.4,.33);
\draw[draw=black](1.3,.26)--(1.43,.29);

\filldraw[draw=white,fill=black!20,fill opacity=0.5](0,0)--(1.8,.-.7)--(1.8,.8)--(0,1.5)--cycle;
\draw[arrows=-,color=gray](0,0)--(0,.936);
\draw[arrows=->,color=gray](0,.936)--(0,2.339);

\draw[draw=black,very thick,->](0,0)--(1.3,.2);
\draw[draw=black,dashed](0,.7)--(1.3,.2);
\draw[draw=black,very thick, ->](0,0)--(0,.7);
\draw[draw=black,dashed](1.3,.2)--(1.3,.-.5);
\draw[draw=black,dashed](0,0)--(1.3,.-.5);
\filldraw[color=black](1.3,.-.5) circle (2pt);

\path (2.126,-.534) node[below] {$\color{black}y(\jmath)$}
	       (0,2.339) node[right] {$\color{black}z(\kappa)$}
	       (-1.432,-.85) node[below] {$\color{black}x(\imath)$};
\path (0,.7) node[left] {$\color{black}\bb'$$=\bb_\kappa$};
\path (1.55,.37) node[above] {$\color{black}\bb''$$=\bb_\jmath+\bb_\kappa$};

\path (.9,-.5) node[below]{$\color{red}\tilde \ba$};
\path (1.3,.2) node[right]{$\color{black}\tilde \bb$};
\path (-1.33,-.78) node[above]{$\color{blue}\tilde \bc$};
\end{tikzpicture}
\caption{Assume the channel outputs of three users' signal vectors: the back vector $\tilde \bb$ is the desired signal; the red vector $\color{red}\tilde \ba$ is the self-interference; the blue vector $\color{blue}\tilde \bc$ is the other interference signal.}
\label{phy_explain}
\end{figure}

\subsection{Network-coded Relaying}

The MIMO switch described in Section II.A makes use of zero-forcing relaying, by which data are switched based on a permutation matrix, whose diagonal elements are all zero. If the $k$th diagonal element is nonzero, it means that the relay forwards the signal from station $k$ back to itself. There is no need to force a diagonal element to zero because the self-interference is known and can be removed. This is the basic idea behind physical-layer network coding (PNC) \cite{Zhang06physical-layernetwork}, which underlies many other works, e.g., \cite{guo10,yuan11,ho09,meixia11}.

In general, allowing PNC improves the performance. Even though the self-interference costs the relay some energy, removing the constraint on the diagonal of the derangement enlarges the possible set of the optimization problem, and thereby yields a better optimal objective. This can also be seen from an example of multi-way relaying, in which three single-antenna stations communicate with the help of a three-antenna relay. The traffic switching pattern among the three stations are defined in Fig.~\ref{3node_realization}. The desired signal of station $1$ is the black signal $b$; the red signal $\color{red}a$ is its self-interference; the blue signal $\color{blue}c$ is the other interference signal. We assume the three signals and the channel gains are real-valued for simpler illustration. A sketch of the three signals after passing through the channel is shown in Fig.~\ref{phy_explain}, i.e., $\color{red}\tilde \ba$, $\color{black}\tilde \bb$ and $\color{blue}\tilde \bc$, which are all three-dimensional vectors due to three antennas at the relay. As shown in Fig.~\ref{phy_explain}, we assume that $\color{blue}\tilde \bc$ is along  $x$-axis; $\color{red}\tilde \ba$ is in $xy$-plane; and $\tilde \bb=\bb_\imath+\bb_\jmath+\bb_\kappa$. In the case of zero-forcing, the desired signal of station $1$, $\tilde \bb$, should be projected to $z$-axis, which is perpendicular to $xy$-plane spanned by the two signals $\color{red}\tilde \ba$ and $\color{blue}\tilde \bc$. In this way both the two interference signals are zeroed out and we get a post-processing signal $\color{black}\bb'$$=\bb_\kappa$. However, with network coding, we do not have to zero out the self-interference. By dropping this constraint, we only need to project the desired signal to the $yz$-plane which is perpendicular to the interference signal $\color{blue}\tilde \bc$, then we obtain the projected signal $\color{black}\bb''$$=\bb_\jmath+\bb_\kappa$. Obviously, the projection $\color{black}\bb''$ with network coding is stronger than $\color{black}\bb'$ by zero-forcing.

With PNC, we rewrite (\ref{4}) as
\begin{equation} \label{formula_G2}
\bH^T\bG\bH=\bA(\bP+\bB),
\end{equation}
where $\bB = \diag\{b_1, \cdots, b_N\}$ is a diagonal matrix to be determined.

As an example of a symmetric derangement, the corresponding network-coded switch matrix has the pattern:
\be \label{ps}
\bP_1 + \bB =
 \left[ {\begin{array}{*{20}c}
   b_1 & 0 & 0 & 1  \\
   0 & b_2 & 1 & 0  \\
   0 & 1 & b_3 & 0  \\
   1 & 0 & 0 & b_4  \\
\end{array}} \right].
\ee
This switching pattern corresponds to two pairwise data exchanges, in which stations $1$ and $4$ are one pair and stations $2$ and $3$ are the other pair. Hence, the network-coded MIMO switching can construct multiple parallel two-way relay transmissions. An example of asymmetric derangement may have the following switch matrix:
\be  \label{pa}
\bP_2+\bB =
 \left[ {\begin{array}{*{20}c}
   b_1 & 0 & 1 & 0  \\
   1 & b_2 & 0 & 0  \\
   0 & 0 & b_3 & 1  \\
   0 & 1 & 0 & b_4  \\
\end{array}} \right].
\ee
This generalizes the traditional physical-layer network coding setting as presented in \cite{Zhang06physical-layernetwork} because the data exchange is not pairwise. For both symmetric and asymmetric switch matrices, we shall refer to the corresponding matrices with nonzero diagonal as network-coded switch matrix, and the associated setup as MIMO switching with network-coded relaying.

Once the receivers compensate for the amplification and remove  self-interference, the resulting signals form this vector:
\be
\hat \br = \bP\bx+\bv',
\ee
where the post-processing noise $\bv'$ is expressed as
\begin{align}
\bv' = (\bP+\bB)\bH^{-1}\bu+\bA^{-1}\bw.
\end{align}
The covariance of $\bv'$ is written as
\be
\mathbb{E}\{\bv'\bv'^H\}=\bR -\bI + \sigma^2 \bA^{-1}\bA^{-H},
\ee
where
\be \label{def_R}
\bR \triangleq \bI + \gamma^2(\bP+\bB)\bH^{-1}\bH^{-H}(\bP+\bB)^H.
\ee
The constraint of the relay power consumption is rewritten as
\be
&&\kern-3.3em \Omega (\bA,\bB) \notag\\
&&\kern-3.3em \triangleq \text{Tr}[\bH^{-T} \kern-.3em \bA (\bR+\bB\bP^T \kern-.3em+\bP\bB^H+\bB\bB^H)\bA^H (\bH^*)^{-1} ] \\
&&\kern-3.3em =\text{Tr}[\bA (\bR+\bB\bP^T \kern-.3em+\bP\bB^H \kern-.3em +\bB\bB^H)\bA^H (\bH^*)^{-1} \bH^{-T}] \label{def_O}\\
&&\kern-3.3em \leq  p. 
\ee

We have thus established another framework for MIMO switching by network-coded relaying method, in which $\bP+\bB$ is a switch matrix. This framework can be generalized to the case where the switch matrix realizes a general transmission pattern. For example, if there are two nonzero non-diagonal elements in a column of the switch matrix, then a multicast connection is being realized within one switch matrix. In fact, by scheduling a set of switch matrices, each realizing a permutation, we can satisfy arbitrary user traffic patterns.

\section{The Maxmin Problem}

In this section, we formulate a \emph{maxmin problem}, in which the minimum received SNR among all the stations is maximized. According to (\ref{8}), the post-processing noise power of receiver $i$ is
\be
\epsilon_i= q_i - 1 +\frac{\sigma^2}{|a_i|^2},
\ee
where $q_i\triangleq\bQ_{ii}$. Thus, the received SNR is ${1}/{\epsilon_i}$. Since the system is half-duplex with uplink and downlink of equal duration, the throughput achieved by Gaussian signaling is
\be
c_i = \frac{1}{2}\log_2 \left(1+\frac{1}{\epsilon_i}\right),
\ee
in bits per symbol period.

We first study the maxmin problem without PNC in Section III.A, and then allow PNC in Section III.B.

\subsection{Zero-forcing Relaying}

Let $\epsilon$ denote the maximum post-processing noise power among the stations. An optimization problem is formulated as follows:
\begin{subequations} \label{maxmin}
\begin{align}
\min \limits_{\ba}  &\qquad \epsilon  \label{22a}\\
\text{s.t.} &\qquad  |a_i|^2 \ge \frac{\sigma^2}{\epsilon + 1 - q_i},\ i=1,\cdots,N, \label{22b}\\
 &\qquad \ba^H \bS \ba\le p, \label{22c}\\
 &\qquad \epsilon \ge 0.
\end{align}
\end{subequations}

\begin{lemma}  \label{lem_eq2}
Every optimal solution for the optimization problem (\ref{maxmin}) must satisfy the relay power constraint (\ref{22c}) with equality.
\end{lemma}

\begin{proof}
Let $\tilde \ba=[\tilde a_1,\cdots,\tilde a_N]^T$ denote the optimal solution for the problem (\ref{maxmin}) with the optimal objective $\tilde \epsilon$.
Suppose $\tilde \ba$ satisfies the constraint \eqref{22c} with strict inequality. Then there exists $\tau>0$ such that
\be
\ba^H \bS \ba < p,
\ee
for all $\ba$ with $|a_i|\in(|\tilde a_i|-\tau,|\tilde a_i|+\tau),\ i=1,\cdots,N$. Let $|\tilde a_i'| = |\tilde a_i| + \frac{\tau}{2},\ i=1,\cdots,N$, and let 
\be
\Xi_i(a_i) \triangleq q_{i}+\frac{\sigma^2}{|a_i|^2}.
\ee
Then
\be
\Xi_i(\tilde a_i') < \Xi_i(\tilde a_i),\quad i=1,\cdots,N.
\ee
Since each $\Xi_i(\tilde a_i')$ is smaller than $\Xi_i(\tilde a_i)$ for all $i$, the maxmin objective becomes smaller with the solution $\tilde \ba'=[a_1',\cdots,a_N']^T$. Thus, $\tilde \ba$ is not the optimal solution, and contradiction arises.
\end{proof}

With this lemma, we could eliminate the feasible solutions with which the relay consumes less power than $p$.

\begin{proposition} \label{lem_np}
The maxmin problem\footnote{Essentially, it is a minmax problem with respect to (w.r.t.) the post-processing noise power. In this paper, we call it the maxmin problem w.r.t. the received SNR to keep it consistent with the following equal-SNR problem proposed in Section IV.} \eqref{maxmin} is equivalent to the following quadratically constrained quadratic program (QCQP),
\begin{subequations}  \label{maxmin2}
\begin{align}
\min \limits_{\ba} &\qquad \ba^H \bS \ba \label{43a}\\
\text{s.t.} &\qquad |a_i|^2 \ge \frac{\sigma^2}{\tilde \epsilon+1 - q_i},\ i=1,\cdots,N. \label{43b}
\end{align}
\end{subequations}
\end{proposition}

\begin{proof}
Let $\tilde \ba$ be an optimal solution of (\ref{maxmin}), and the associated optimal objective is $\tilde \epsilon\ge 0$. Note that $\tilde \epsilon$ is the maximum noise power for all $i$.

The solution $\tilde \ba$ satisfies \eqref{43b}, where the largest noise power among all $i$ is equal to $\tilde \epsilon$. Hence, $\tilde \ba$ is a feasible solution of (\ref{maxmin2}). According to Lemma \ref{lem_eq2}, the power consumption of the relay is $\tilde\ba^H \bS \tilde\ba=p$. Let $\hat \ba$ be an optimal solution of (\ref{maxmin2}). The power consumption of the relay can not be larger than $p$, otherwise $\hat \ba$ is even worse than $\tilde \ba$ for (\ref{maxmin2}). If the power consumption of the relay $\hat \ba^H\bS\hat \ba$ is strictly smaller than $p$, then $\hat \ba$ is a feasible solution of (\ref{maxmin}), which is at least as good as $\tilde \ba$ since the maximum noise power is not larger than $\tilde \epsilon$. That is, $\hat \ba$ is an optimal solution of (\ref{maxmin}). However, the power consumption $\hat \ba^H \bS \hat \ba<p$ contradicts Lemma \ref{lem_eq2}. Thus, $\hat \ba^H \bS \hat \ba=p$. Furthermore, there is at least one constraint in (\ref{maxmin2}) in which equality holds for $\hat \ba$, otherwise $\hat \ba$ is a better solution for (\ref{maxmin}).
We have proved the optimal solution of (\ref{maxmin}) is also the optimal solution of (\ref{maxmin2}), vice versa. Therefore, the two problems are equivalent.
\end{proof}

\begin{proposition} \label{lem_np}
The maxmin problem (\ref{maxmin}) is NP-hard in the size of $N$ when $\gamma \neq 0$.
\end{proposition}

\begin{proof}
Problem (\ref{maxmin2}) is equivalent to problem (2) in \cite{luo06}, which has been proved to be NP-hard in general. Thus, problem \eqref{maxmin} is NP-hard as well.  We do not repeat the steps in \cite{luo06}.
\end{proof}

We can use the semidefinite relaxation (SDR) technique in \cite{luo06} to find a throughput upper bound and a suboptimal solution for (\ref{maxmin2}). Let $\bX=\ba\ba^H$, then (\ref{maxmin2}) can be rewritten as
\begin{subequations}  \label{maxmin3}
\begin{align}
\min_{ \mathclap{\boldsymbol{X} \in \mathbb{C}^{N\times N},\ell_i}} &\qquad \text{Tr}[\bS\bX] \label{44a}\\
\text{s.t.} &\qquad \text{Tr}[\bE_i\odot\bX] - \ell_i = \frac{\sigma^2}{\tilde \epsilon + 1 - q_i},\\
&\qquad  \ell_i\ge 0,\ i=1,\cdots,N,\\
&\qquad \bX \succeq {\bf 0},\\
&\qquad \text{rank}(\bX)=1, \label{28e}
\end{align}
\end{subequations}
where $\odot$ denotes element-by-element multiplication, i.e., the Hadamard product; $\bX \succeq {\bf 0}$ means the matrix $\bX$ is symmetric positive semidefinite; $\bE_i$ is an $N\times N$ matrix, in which element $(i,i)$ is $1$ and all the other elements are $0$; $\ell_i,\ i=1,\cdots,N$ are ``slack'' variables. If we drop the rank-$1$ constraint (\ref{28e}), problem (\ref{maxmin3}) is in the standard form of a semidefinite programming problem (SDP).

\medskip\noindent{\bf\emph{Overview of Our Approach to Maxmin Optimization Problem}}: We now overview the approach of our numerical investigation of the maxmin optimization problem. Although we focus on zero-forcing relaying here, we use the same approach for the study of maxmin optimization for network-coded relaying as well after its corresponding formulation is set up in Part B.

We recognize that it is difficult to solve the maxmin problem \eqref{maxmin}. Therefore, we attempt to find its suboptimal solution. The logical steps and the rationale for our approach to finding a suboptimal solution are summarized below:
\begin{enumerate}
\item We have already proved that the maxmin problem is equivalent to the QCQP \eqref{maxmin2}. Unfortunately, solving for the optimal solution of the QCQP is still difficult.
\item Fortunately, the SDR technique proposed in \cite{luo06} can be used to find a good suboptimal solution of the QCQP. The technique consists of the following three steps:
\begin{enumerate}
\item Rewrite the QCQP as an optimization problem \eqref{maxmin3} with a constraint of rank $1$.
\item By dropping the constraint of rank $1$, the problem \eqref{maxmin3} is turned into a SDP. The optimal solution of the SDP can be found by some toolboxes, such as SeDuMi and SDPT3. Importantly, the optimal objective of the SDP is an upper bound of that of the QCQP, since the QCQP has one extra constraint, i.e., the constraint of rank $1$.
\item We then use the randomization technique in \cite{luo06} to approximate the optimal solution of the QCQP based on the optimal solution of the SDP. With the randomization technique, we get a result, which satisfies the constraints of the QCQP, and this result is at least a suboptimal solution of the QCQP.
  \end{enumerate}
\end{enumerate}

In \cite{luo06}, the authors showed that the suboptimal solution obtained as above can achieve an objective close to the global optimum. As will be shown,
our simulation results also validate the near-optimality statement in \cite{luo06}.

\medskip\noindent{\bf\emph{Numerical Method (One-dimensional Search)}}: As per the discussion in the above overview, we use the suboptimal solution of  \eqref{maxmin3} to approximate the solution of \eqref{maxmin}. The optimal solution of \eqref{maxmin} is the value of $\tilde \epsilon$ in \eqref{maxmin3}, for which the objective of \eqref{maxmin3}, i.e., the minimum power consumption is $p$.
We solve for $\tilde \epsilon$ by one-dimension search from $\max\limits_{j}q_j-1$, i.e., when the power consumption of the relay (\ref{44a}) approaches infinity. In each step of the search, $\tilde \epsilon$ is increased by a small amount $\tilde \delta$. Given an $\tilde \epsilon$, we can use SDP solvers to find the minimum power consumption of the relay. When the step size $\tilde \delta$ is small enough, the first value of $\tilde \epsilon$, for which the minimum power consumption is $p$, is the solution \eqref{maxmin}.


\subsection{Network-coded Relaying}

The network-coded problem is formulated as
\begin{subequations} \label{netcod}
\begin{align}
\min \limits_{\boldsymbol{A},\boldsymbol{B}}  &\qquad \epsilon  \label{28a}\\
\text{s.t.} &\qquad  \bR_{ii}-1+\frac{\sigma^2}{|a_i|^2}\leq \epsilon,\ i=1,\cdots,N, \label{28b}\\
 &\qquad \Omega(\bA,\bB)\leq p, \label{28c}\\
 &\qquad \epsilon \ge 0,
\end{align}
\end{subequations}
where $\bR$ and $\Omega(\bA,\bB)$ are defined in \eqref{def_R} and \eqref{def_O}, respectively.
In contrast to the previous problem of zero-forcing relaying, the constraint of the relay power consumption \eqref{28c} is quartic, making it more difficult than \eqref{maxmin}. We propose an iterative algorithm, in which $\bA$ and $\bB$ defined in \eqref{formula_G2} are optimized iteratively.

\subsubsection{Optimize $\bA$ for given $\bB$}
Given $\bB$, problem \eqref{netcod} can be formulated as \eqref{maxmin} by redefining
\begin{align}
s_{ij} &\triangleq [\bR+\bB\bP^T+\bP\bB^H+\bB\bB^H]_{ji} [(\bH^{T}\bH^*)^{-1}]_{ij},\\
q_{i} &\triangleq \bR_{ii},
\end{align}
Thus, the network-coded problem with fixed $\bB$ can be solved by the SDR technique, which is the same as solving (\ref{maxmin2}).

\subsubsection{Optimize $\bB$ for given $\bA$}
For ease of notation, let
\bal
\bM\triangleq&\bA^H (\bH^*)^{-1} \bH^{-T} \bA,\\
\bN\triangleq&\bI+\gamma^2\bH^{-1}\bH^{-H}.
\end{align}
The relay power consumption \eqref{def_O} is rewritten as
\begin{align}
\Omega(\bA,\bB)=&\text{Tr}[(\bP+\bB)\bN(\bP+\bB)^{H} \bM]\\
=&\text{Tr}[\bB\bN\bB^{H}\bM+\bB\bN\bP^{T}\bM  \notag\\
&+\bM\bP\bN\bB^{H}+\bP\bN\bP^{T}\bM]. \label{29}
\end{align}
Constraint \eqref{28b} can be rewritten as
\be  \label{30}
\left[\bI + \gamma^2 (\bP+\bB)\bH^{-1}\bH^{-H}(\bP+\bB)^{H}\right]_{ii}
\leq  \epsilon -\frac{\sigma^2}{|a_i|^2}.
\ee
Both the objective and the constraint w.r.t. $\bB$ are inhomogeneous quadratic. To homogenize this problem, we let $\tilde \bb = [\bb^T,t]^T$, where $\bb = \diag\{\bB\}$. Then the relay power consumption \eqref{29} can be written as
\be
\begin{aligned}
&\begin{bmatrix}
\bb^H & t^*
\end{bmatrix}
\begin{bmatrix}
\bS_t & \bff\\
\bff^H & 0
\end{bmatrix}
\begin{bmatrix}
\bb \\ t
\end{bmatrix}+\text{Tr}[\bP\bN\bP^{T}\bM]\\
&\triangleq  \tilde \bb^H \tilde \bS \tilde \bb+\text{Tr}[\bP\bN\bP^{T}\bM],
\end{aligned}
\ee
where $\bS_t = \bN^T\odot\bM$, $\bff = \diag\{\bM\bP\bN\}$, $t=1$.
Let $\tilde \bX = \tilde \bb \tilde \bb^H$, then this problem can be homogenized as
\begin{subequations}  \label{maxmin4}
\begin{align}
\min_{\mathclap{\boldsymbol{\tilde X}\in \mathbb{C}^{N\times N},\ell_i}} &\quad \text{Tr}[\tilde \bS \tilde\bX]
\label{36a} \\
\text{s.t.} &\quad \text{Tr}[\tilde\bE_i\odot \tilde\bX] + \ell_i \\
&\quad = \frac{1}{\gamma^2}\left(\tilde \epsilon - \frac{\sigma^2}{|a_i|^2}\right)-1-\left[\bP(\bH^H\bH)^{-1}\bP^{H}\right]_{ii},\notag\\
&\quad  \ell_i\ge 0,\ i=1,\cdots,N,\\
&\quad \text{Tr}[\tilde\bE_{N+1}\odot \tilde\bX] = 1, \label{36d}\\
&\quad \bX \succeq {\bf 0},\\
&\quad \text{rank}(\bX)=1, \label{36f}
\end{align}
\end{subequations}
where $\tilde \bE_i,\ i=1,\cdots,N$, can be easily designed according to \eqref{30}. The last element of $\tilde \bE_{N+1}$ is $1$ and all the others are $0$. Thus,  constraint \eqref{36d} is actually $|t|^2 = 1$. The problem can be solved in the same way as solving \eqref{maxmin3}. Note that  constraint \eqref{36d} only requires the magnitude of $t$ to be $1$. Assume that the solution of \eqref{maxmin4} is $\tilde\bb_{opt} = [\hat\bb^T,\hat t]^T$, then $\tilde \bb_{opt} e^{j\alpha}$, for all $\alpha$ also guarantees the optimality of \eqref{maxmin4}, since adding a phase rotation does not change $\tilde \bX$. Thus, the solution of this problem is $\hat \bb e^{-j\angle \hat t}$.

To sum up the iterative algorithm, in order to design $\bA$ and $\bB$, we can first initialize $\bA$ or $\bB$ by any diagonal matrix. Then we iteratively perform the two solvers to optimize $\bA$ and $\bB$ until convergence.


\section{The Equal-SNR Problem}

In this section, we require the received SNR at all stations to be identical and find its maximum of the \emph{equal-SNR problem}. Such a formulation was first investigated in \cite{allerton11}. The equal-SNR requirement not only guarantees perfect fairness, but also allows more efficient computation.

\subsection{Zero-forcing Relaying}

The optimization problem is formulated as
\begin{subequations}  \label{esnr}
\begin{align}
\min \limits_{\boldsymbol{a}} &\qquad \epsilon \\
\text{s.t.} &\qquad |a_i|^2 = \frac{\sigma^2}{\epsilon + 1 - q_i},\ i=1,\cdots,N, \label{14b}\\
 &\qquad \ba^H \bS \ba\le p, \label{14c}\\
 &\qquad \epsilon \ge 0.
\end{align}
\end{subequations}
The only difference between the equal-SNR problem and the maxmin problem \eqref{maxmin} is that the inequality \eqref{22b} is replaced by the equality \eqref{14b}. In \cite{wms11}, we only provided suboptimal solutions. In this paper, we will investigate optimal solutions and propose analytical and numerical suboptimal solutions. As has been proved in \cite{wms11}, the equal-SNR problem is feasible.  Another property is that the solution of the equal-SNR problem is feasible for the maxmin problem, since constraint \eqref{22b} of the maxmin problem has a larger possible set than constraint \eqref{14b} in the equal-SNR problem. Thus, the optimal objective of (\ref{esnr}) can not be smaller than that of (\ref{maxmin}).

\begin{proposition} \label{lem_nogap}
If no additional noise is introduced at the switch, then the optimal solution of (\ref{maxmin}) is such that each station has exactly the same post-processing noise power. That is, in this case, the two optimization problems (\ref{maxmin}) and (\ref{esnr}) are equivalent.
\end{proposition}

\begin{proof}
As proved in Proposition \ref{lem_np}, the optimization problem (\ref{maxmin2}) is equivalent to (\ref{maxmin}).
When $\gamma=0$ in (\ref{maxmin2}), $q_i=1$ for all $i=1,\cdots,N$. The objective can be rewritten as $\sum\nolimits_{i=1}^{N} s_{ii}|a_i|^2$. The optimal solution of $\ba$ is obvious, and it satisfies $|a_i|^2=\frac{\sigma^2}{\tilde \epsilon-1}$. That is, for the optimal solution of (\ref{maxmin}), each station has equal post-processing $\tilde \epsilon$.
\end{proof}

Thus, we can use the solution of the equal-SNR problem to approximate that of the maxmin problem in the high SNR regime. The gap will be evaluated numerically in Section V.

\medskip\noindent{\bf\emph{\underline{Optimal Solution in the Case of Two Stations}}}

In order to minimize the post-processing noise power $\epsilon$, we should try to maximize $|a_1|$ and $|a_2|$. Given any $\epsilon$, the amplitudes $|a_1|$ and $|a_2|$ can be calculated by the equal noise power constraint. Then we should find their optimal phases to minimize the relay power consumption.

If $N=2$, the power constraint of the relay in \eqref{14c} can be expanded as
\be   \label{16}
\begin{aligned}
\ba^H \bS \ba
= s_{11}|a_1|^2+s_{22}|a_2|^2+s_{21} a_1 a_2^* +s_{12} a_1^* a_2,
\end{aligned}
\ee
where $s_{11}\ge0$, $s_{22}\ge0$, and $s_{12}=s_{21}^*$. By the definition of $\bS$ in \eqref{def_S}, $s_{12}=s_{21}\ge 0$. Formula (\ref{16}) can be written as
\be  \label{26}
\ba^H \bS \ba
= \left|\sqrt{s_{11}}a_1+\frac{s_{12}}{\sqrt{s_{11}}}a_2\right|^2+ \left(s_{22}-\frac{s_{12}^2}{s_{11}}\right) |a_2|^2.
\ee
Since in \eqref{26} only $\left|\sqrt{s_{11}}a_1+\frac{s_{12}}{\sqrt{s_{11}}}a_2\right|$ is related to the phases of variables $a_1$ and $a_2$, the global minimum is achieved by real-valued $a_1$ and $a_2$ with opposite signs.  Without loss of optimality, assume $a_1\ge 0$ and $a_2\leq 0$. Then the power consumption constraint can be simplified as
\be
\ba^T \bS \ba
 =  s_{11}a_1^2+2s_{12}a_1a_2+s_{22}a_2^2.
\ee
According to the equal noise power constraint \eqref{14b}, we have
\be  \label{57}
a_2 = -\frac{\sigma}{\sqrt{q_1-q_2+\frac{\sigma^2}{a_1^2}}}.
\ee
Plugging (\ref{57}) into the power constraint, we have a biquartic equation (\ref{65}),
\be  \label{65}
\begin{aligned}
0=& s_{11}^2q_\delta^2 a_1^8   +2q_\delta [s_{11}^2 \sigma^2 + s_{11} s_{22}\sigma^2-
  s_{11}p   q_\delta \\
  &  - 2s_{12}^2\sigma^2] a_1^6 +[\sigma^4 (s_{11}+s_{22})^2 +
   p^2    q_\delta^2 \\
&    -  2p\sigma^2q_\delta(2s_{11}+s_{22}) -4 s_{12}^2\sigma^4] a_1^4 \\
   &  +  [2p^2\sigma^2q_\delta-2p\sigma^4(s_{11}+s_{22})] a_1^2 + p^2\sigma^4,
\end{aligned}
\ee
where $q_\delta=q_1-q_2$.
Since the equal-SNR problem is feasible, there exists solutions for (\ref{65}). In order to maximize $a_1$ and $a_2$, the largest real root of (\ref{65}) is the optimal solution of $a_1$, which admits an analytical solution \cite{wiki}. Consequently, $a_2$ can be calculated by (\ref{57}).

Alternatively, after we deduce that $a_1$ and $a_2$ have opposite signs, the one-dimensional search in Section III.A can be used to solve the problem as well. Thus, we have solved the equal-SNR problem of (\ref{esnr}) in the case of $N=2$.


It is not difficult to see that making the two user signals have opposite signs (phases) minimizes the relay power consumption. The optimal solution of the maxmin problem when $N=2$ also has the opposite-phase property.
As we shall see, in the case of more than two users with users forming pairs, assigning each pair of users opposite phases is an effective scheme. Thus, in general the opposite-phase setting is effective for pairwise transmission in any problem which needs to minimize the relay power consumption, including the preceding maxmin problem.

\medskip\noindent{\bf\emph{\underline{Suboptimal Solution in the Case of $N>2$}}}

If there are more than two stations, we propose a suboptimal solution to \eqref{esnr}. The vector $\ba$ consists of arbitrary complex numbers. Given $\epsilon$, one can obtain $|a_j|$ from \eqref{14b}. Let the phases $\theta_j=\angle a_j,\ j=1,\cdots,N$, be fixed. We solve \eqref{esnr} to obtain the minimum noise variance $\epsilon(\theta_1,\cdots,\theta_N)$, which satisfies $\ba^H\bS\ba=p$.

We now consider the general case of complex-valued $\ba$. With the numerical method of one-dimensional search, the optimization of $a_j$'s amplitudes and phases can be decoupled. There exists an $\epsilon$ such that $\ba^H\bS\ba=p$. Denote such an $\epsilon$ by $\epsilon(\theta_1,\cdots,\theta_N)$ since in each step we regard the amplitudes as constant values. It then suffices to solve for
\begin{equation} \label{minsigma}
\epsilon^{*}=\arg\min\limits_{\theta_1,\cdots,\theta_N} \epsilon(\theta_1,\cdots,\theta_N).
\end{equation}

In the following we provide two suboptimal algorithms for (\ref{esnr}) via \eqref{minsigma}.

\medskip\noindent{\bf\emph{Non-PNC Random-phase Algorithm}}: In (\ref{minsigma}), we note that $\epsilon$ is a complicated nonlinear function of $\theta_j$. A time-consuming exhaustive search can be used to find the solution to (\ref{minsigma}). To reduce the computation time, we find the best set of phases over a randomly generated candidates in lieu of an exhaustive search. We call this the \emph{random-phase algorithm}. We divide the interval of $[0,2\pi)$ equally into $M$ bins with the values of $0,\frac{2\pi}{M},\cdots,\frac{2(M-1)\pi}{M}$ respectively, and we randomly pick among them to set the the value of $\theta_j$ for each and every $j=1,\cdots,N$. After that, we compute the corresponding $\epsilon(\theta_1,\cdots, \theta_N)$ by solving (\ref{14b}) and (\ref{14c}). Given an $\epsilon$, substituting it into (\ref{14b}) yields $|a_j|$ for all $j$. We perform $L$ trials of these random phase assignments to obtain $L$ phase vectors of $(\theta_1,\cdots, \theta_N)$. Calculating the relay power consumption by \eqref{14c}, we choose the phase vector using the least power among the $L$ candidates as our approximated phase solution. By one-dimensional search, the estimated $\epsilon^*$ can be achieved when the least power consumption of the relay is $p$. Accordingly, we can calculate $|a_j|$ for all $j$ with the estimated $\epsilon^*$, then $\ba$ with the approximated phase solution. Hence, $\bG$ can be calculated by its definition \eqref{10}.
This best-out-of-$L$-trials feasible solution is in general larger than the actual optimal $\epsilon^*$. In Section V, we will show that large gains can be achieved with only small $M$ and $L$. Moreover, increasing $M$ and $L$ further yields very little improvement, suggesting that the estimated $\epsilon^*$ with small $M$ and $L$ is close to the result achieved by an extensive search.

\medskip\noindent{\bf\emph{Non-PNC Opposite-phase Algorithm}}: Recall that in the case of $N=2$ the solution has $a_1$ and $a_2$ with opposite signs, or in general the minimum power consumption is achieved with $a_1$ and $a_2$ being complex numbers with opposite phases, i.e., $\theta_1 = \theta_2 + \pi$. For large even number of $N$, consider the situation in which the transmissions are pairwise. The stations form pairs, and two stations in a pair exchange data with each other only. Assume that station pair $\ell$ consists of stations $\pi(\ell)$ and $\kappa(\ell)$. Define
\be  \label{hh}
\begin{bmatrix}
  h_{\ell1}   & \kern-.9em h_{\ell2}  \\
   h_{\ell2}^* &\kern-.9em  h_{\ell3}  \\
\end{bmatrix}
\triangleq
\begin{bmatrix}
[(\bH^{H}\bH)^{-1}]_{\pi(\ell),\pi(\ell)}   & \kern-.9em   [(\bH^{H}\bH)^{-1}]_{\pi(\ell),\kappa(\ell)}  \\
   [(\bH^{H}\bH)^{-1}]_{\kappa(\ell),\pi(\ell)}   & \kern-.9em   [(\bH^{H}\bH)^{-1}]_{\kappa(\ell),\kappa(\ell)}  \\
\end{bmatrix}
\ee
where $h_{\ell1}\ge 0$, $h_{\ell3}\ge 0$, $h_{\ell2}\in \mathbb{C}$. The post-processing noise power can be factorized in terms of $b_{\ell1}$ and $b_{\ell2}$ as follows:
\begin{align}
\epsilon =& \gamma^2 h_{\ell3} + \frac{\sigma^2}{|a_{\ell1}|^2}\\
=&\gamma^2 h_{\ell1} + \frac{\sigma^2}{|a_{\ell2}|^2}. \label{82}
\end{align}
The relay power consumption can be written as
\be \label{omega1}
\begin{aligned}
  \begin{split}
 \Omega_1 (\ba) = & \gamma^2\mathcal{O}_1(\ba)+\sum\limits_{\ell=1}^{N/2}\left\{\left|a_{\ell1} + \gamma^2\left|h_{\ell2}\right|^2  a_{\ell2} \right|^2 \right.\\
 & + \left(h_{\ell1}+ \gamma^2 h_{\ell1} h_{\ell3} -1 \right) |a_{\ell1}|^2 \\
  &\left.+ \left(h_{\ell3}+ \gamma^2 h_{\ell1} h_{\ell3} - \gamma^4|h_{\ell2}|^4 \right) |a_{\ell2}|^2\right\},
 \end{split}
\end{aligned}
\ee
where the summation consists of the inner-pair quadratic items of $\ba$, and
$\mathcal{O}_1(\ba)$ denotes the sum of the quadratic items of $\ba$ across pairs, e.g., $a_ia_j$ where $i,\ j$ from different pairs. For fixed amplitudes of $a_{\ell1}$ and $a_{\ell2}$, it is obvious that the phase vector $(\theta_1,\cdots,\theta_N)$ that minimizes \eqref{omega1} is the optimal phase vector of \eqref{esnr}. We could let two stations of a pair to have opposite phases to lower the inner-pair part of \eqref{omega1}. Simulation results indicate that as long as the relative phase is $\pi$ within a pair, the throughput performance remains essentially the same regardless of the phase differences between different pairs. Thus, to simplify the problem we use real numbers for the elements of $\ba$. The amplitudes are calculated in the same way as that of the non-PNC random-phase algorithm.

\begin{remark}
The opposite-phase solution approaches the optimal solution of the equal-SNR problem as SNR increases, since $\gamma^2\mathcal{O}_1(\ba)$ becomes negligible. Thus, the opposite-phase solution approaches the optimal solution of the maxmin problem as well according to Proposition $3$.
\end{remark}

\subsection{Network-Coded Relaying}

The optimization problem is formulated as
\begin{subequations} \label{netcod2}
\begin{align}
\min \limits_{\boldsymbol{A},\boldsymbol{B}}  &\qquad \epsilon  \label{282a}\\
\text{s.t.} &\qquad  \bR_{ii}-1+\frac{\sigma^2}{|a_i|^2} = \epsilon,\ i=1,\cdots,N, \label{282b}\\
 &\qquad \Omega(\bA,\bB)\leq p, \label{282c}\\
 &\qquad \epsilon \ge 0.
\end{align}
\end{subequations}
The only difference between the equal-SNR problem and the maxmin problem \eqref{netcod} is that the inequality \eqref{28b} is replaced by the equality \eqref{282b}. A proposition for the network-coded problem can be proved in the same way as Proposition $3$. We then use the solution of the equal-SNR problem \eqref{netcod2} to approximate that of the maxmin problem \eqref{netcod} at high SNR. The gap will be also evaluated in Section V.

%
%
%


\medskip\noindent{\bf\emph{PNC Identical-$b$ Random-phase Algorithm}} simplifies the matter by introducing an extra constraint:  $\bB = b\bI$, where $b$ is a real scalar. We set a range for searching $b$. For each trial of $b$, we use essentially the same method as the non-PNC random-phase algorithm proposed in Section IV.A to find the phases $\theta_j$ of $a_j$ and $\varphi_j$ of $b_j$, and find the corresponding estimated $\epsilon^{*}$ as well. We then find $b$ that yields the least post-processing noise power. The beamformer $\bG$ can be calculated from $\bP$, $|a_j|,\ b,\ \theta_j$ and $\varphi_j$.

\medskip\noindent{\bf\emph{PNC Phase-aligned Algorithm}}:
Consider pairwise transmission, i.e., the case of $N$ being an even number. Define
\be  \label{hh}
\bB_\ell\triangleq\left[ {\begin{array}{*{20}c}
   b_{\ell1} & 0  \\
   0   & b_{\ell2}  \\
\end{array}} \right],
\ee
where $b_{\ell1}$ and $b_{\ell2}$ are the diagonal elements $\pi(\ell)$ and $\kappa(\ell)$ of $\bB_\ell$, respectively. The post-processing noise power can be factorized in terms of $b_{\ell1}$ and $b_{\ell2}$ as follows:
\begin{align}
\epsilon &=\gamma^2 h_{\ell1} \left|b_{\ell1}+\frac{h_{\ell2}^*}{h_{\ell1}}\right|^2 - \frac{\gamma^2 |h_{\ell2}|^2}{h_{\ell1}} + \gamma^2 h_{\ell3} + \frac{\sigma^2}{|a_{\ell1}|^2}\\
 &=\gamma^2 h_{\ell3} \left|b_{\ell2}+\frac{h_{\ell2}}{h_{\ell3}}\right|^2  - \frac{\gamma^2 |h_{\ell2}|^2}{h_{\ell3}} + \gamma^2 h_{\ell1} + \frac{\sigma^2}{|a_{\ell2}|^2}. \label{82}
\end{align}
This formula shows the potential advantage of introducing $b_{\ell1}$ and $b_{\ell2}$, i.e., the application of physical-layer network coding.
If we set $b_{\ell1}$ and $b_{\ell2}$ as
\be \label{85}
b_{\ell1} = -\frac{h_{\ell2}^*}{h_{\ell1}},\quad b_{\ell2} = -\frac{h_{\ell2}}{h_{\ell3}},
\ee
the post-processing noise power can be minimized in terms of $b_{\ell1}$ and $b_{\ell2}$. The relay power consumption can be written as
\be \label{86}
\begin{aligned}
  \begin{split}
 &\Omega_2(\bA)
  = \gamma^2\mathcal{O}_2(\bA) + \sum\limits_{\ell=1}^{N/2} \left\{ \left|a_{\ell1} + \lambda_\ell  a_{\ell2} \right|^2 \right.\\
  &+\kern-.3em \left(h_{\ell1}+\frac{|h_{\ell2}|^2}{h_{\ell1}} - \gamma^2 |h_{\ell2}|^2
    + \gamma^2 h_{\ell1} h_{\ell3} -1 \right) \kern-.2em |a_{\ell1}|^2 \\
  & \left. + \kern-.3em \left(h_{\ell3}+\frac{|h_{\ell2}|^2}{h_{\ell3}} - \gamma^2 |h_{\ell2}|^2 + \gamma^2 h_{\ell1} h_{\ell3} - \lambda_\ell^2 \right)\kern-.2em |a_{\ell2}|^2 \right\},
 \end{split}
\end{aligned}
\ee
where
\[
\lambda_\ell = |h_{\ell2}|^2 \left[\gamma^2\left(\frac{|h_{\ell2}|^2 }{h_{\ell1} h_{\ell3}} - 1 \right)-\frac{1}{h_{\ell1}}-\frac{1}{h_{\ell3}}\right].
\]
Note that $\mathcal{O}_2(\bA)$ denotes the sum of the quadratic items of $\bA$ across pairs. Since $\bH^{-1}\bH^{-H}$ is positive semidefinite, we can prove $|h_{\ell2}|^2 \leq h_{\ell1} h_{\ell3}$, then $\lambda_\ell<0$.
Thus, for fixed amplitudes of $a_{\ell1}$ and $a_{\ell2}$, the inner-pair part of \eqref{86} can  be lowered when the phases of $a_{\ell1}$ and $a_{\ell2}$ are aligned, e.g., $a_{\ell1},\ a_{\ell2}\ge 0$. Simulation results indicate that as long as the phases of $\bA$'s diagonal elements are aligned within pairs, the throughput performance is not sensitive to the phase differences among different pairs.

\begin{remark}
Note that the phase-aligned solution is the optimal solution of the maxmin problem \eqref{netcod} when the noise level at the relay is zero. The result is obvious after setting $\gamma$ to $0$ in \eqref{netcod}.
\end{remark}

In summary, Remarks $1$ and $2$ show that when the noise at the relay is small, the opposite-phase solution approaches optimal for non-PNC case  and the phase-aligned solution approaches optimal for the PNC case. Our simulation results validate this conclusion.

\section{Numerical Results}
In this section, we evaluate the throughputs of various designs of MIMO switches. We assume the maximum transmit power of the relay and every station are the same (thus $p=1$), and the noise level at the relay and the stations is the same.
Our simulation indicates that the system throughputs are roughly the same with different symmetric permutations. The same result can be concluded for asymmetric permutations. Thus, we use one permutation for each of them ($\bP_1$ and $\bP_2$ are given by the matrices described in (\ref{ps}) and (\ref{pa}) with the diagonal elements set to $0$) in simulations.

\medskip\noindent{\bf\emph{Observation 1}}: For equal-SNR zero-forcing (i.e., non-PNC) relaying, the optimal setting for the case of two stations ($N=2$) has the property that the two elements of $\ba$ have opposite signs. In general, the non-PNC opposite-phase algorithm is an effective scheme for pairwise transmission with larger even number of stations.

A similar framework as ours is investigated in \cite{amah09}, which focuses on optimizing the sum rate of all stations. Therein, a suboptimal beamforming scheme is proposed, which also uses zero-forcing detection and zero-forcing precoding, and simply uses a positive scalar weight to control the relay power consumption instead of our diagonal $\bA$. We regard this scheme as a benchmark and call it ``the basic scheme.'' All schemes proposed in this paper have an advantage over the basic scheme in that they guarantees fairness.

Compared with the basic scheme in Fig.~\ref{2u}, the optimal setting of $N=2$, i.e., the non-PNC opposite-phase algorithm proposed in Section IV.A achieves more than $0.6$ dB gain in the low SNR regime. The gain becomes smaller as the SNR increases, e.g., around $0.25$ dB gain at the SNR of $15$ dB. We explain why the gain diminishes for high SNR as follows. When the relay noise power is zero, $|a_1|=|a_2|$. In this case, $s_{12}$ becomes $0$ in (\ref{26}). Then the throughput performance does not depend on the phase difference of $a_1$ and $a_2$.
The opposite-sign setting is equivalent to the identical-gain setting, i.e., the basic scheme. Thus, the gain over the basic scheme becomes trivial in the high-SNR regime.

Fig.~\ref{4u} presents the throughput in the case where $N=4$ stations form two pairs for pairwise transmission. The throughputs are roughly the same when we vary the phase differences of the pairs while keeping the phase difference within each pair to $\pi$. (Experimentation with the phases is not shown in Fig.~\ref{4u} to avoid cluttering.) With this result, we could set the elements of $\ba$ such that one element in each pair is positive and the other element is negative. The results are similar to that of $N=2$. The non-PNC opposite-phase algorithm achieves $0.8$ dB gain in the low SNR regime and $0.25$ dB gain over the basic scheme in the high SNR regime.

\begin{figure}
\centering
\includegraphics[width=3.5in]{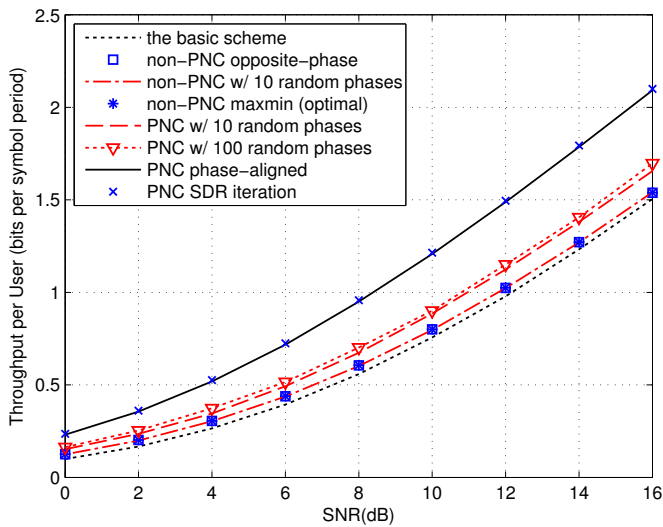}
\caption{Throughput comparison of different relaying schemes in the case of two stations.}
\label{2u}
\end{figure}

\medskip\noindent{\bf\emph{Observation 2}}: For equal-SNR relaying, physical-layer network coding can be applied as a relaying method to improve throughput performance significantly. For pairwise transmission, the PNC phase-aligned algorithm achieves significant gains over any other non-PNC scheme of zero-forcing relaying.

We present simulation results on the PNC phase-aligned scheme in Fig.~\ref{2u} and Fig.~\ref{4u}. When we apply physical-layer network coding in our MIMO switching, significant gains can be achieved over other schemes. The proposed PNC phase-aligned scheme outperforms all the other non-PNC schemes. Note in particular that compared with the basic scheme, it does not involve complicated calculations during the one-dimensional search of $\bA$ and the setting of $\bB$. However, the PNC phase-aligned scheme can not be applied to non-pairwise transmissions.

\medskip\noindent{\bf\emph{Observation 3}}: For equal-SNR zero-forcing relaying with the non-PNC random-phase  algorithm, the simulation results indicate that large gains can be achieved with a small number of phase bins and trials. For network-coded relaying, the PNC identical-$b$ random-phase algorithm can be applied for both pairwise and non-pairwise transmissions to achieve the network coding gain. It is worth mentioning that network coding helps not only for the traditional pairwise switching pattern but also for the non-pairwise pattern.

In Figs.~\ref{2u} and \ref{4u}, when $M=8$ and $L=10$, the non-PNC random-phase scheme proposed in Section IV.A can achieve good enough throughput performance. With the PNC random-phase algorithm proposed in Section IV.B, the throughput performance is even better than that of the best non-PNC scheme. For $M=8$, it achieves around $1.2$ dB gain with $L=10$ and $1.4$ dB with $L=100$ compared to the basic scheme in the case of two stations; it can achieve around $1.5$ dB and $2$ dB gains with $10$ and $100$ trials, respectively,  over the basic scheme in the case of four stations. However, it needs larger $M$, $L$ and more random trials to perform as good as the PNC phase-aligned scheme.

Consider the non-pairwise transmission in Fig.~\ref{4u_np}. Note that the non-PNC opposite-phase scheme and the PNC phase-aligned scheme can not be applied to non-pairwise transmissions. The non-PNC random-phase scheme and the PNC random-phase scheme can still outperform the basic scheme. For $M=8$, when $L=10$ the non-PNC random-phase scheme could achieve around $0.15$ dB gain. The PNC random-phase scheme could achieve around $1.1$ dB and $1.3$ dB gains over the basic scheme with $10$ and $100$ trials, respectively.

\begin{figure}
\centering
\includegraphics[width=3.5in]{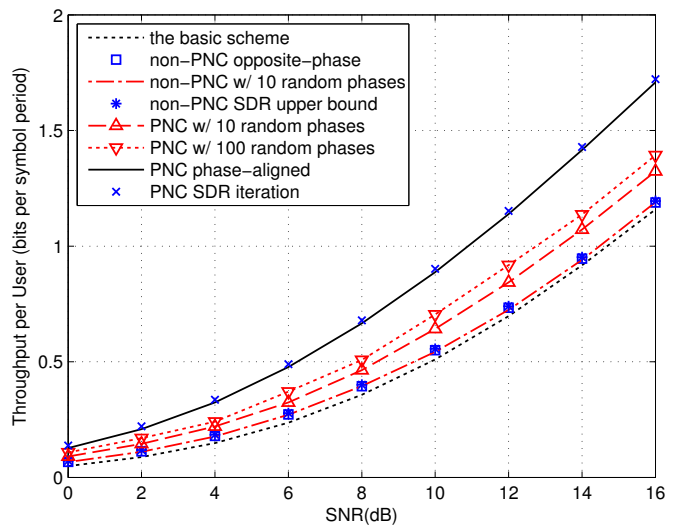}
\caption{Throughput comparison of different relaying schemes for pairwise switching pattern in the case of four stations.}
\label{4u}
\end{figure}

\medskip\noindent{\bf\emph{Observation 4}}: In general, for zero-forcing (i.e., non-PNC) relaying, the throughput of the maxmin problem is larger than that of the equal-SNR problem. However, the throughput gap between the equal-SNR problem and the maxmin problem is small over a wide range of SNR.

We first evaluate the throughput performances with two stations. An interesting question is how large is the throughput gap between (\ref{maxmin}) and (\ref{esnr}) corresponding to their optimal noise powers. In Fig.~\ref{2u}, the throughput gap between the two curves is very small. Note that the optimal solutions of the maxmin problem are found by exhaustive search. In Table I, the gap is evaluated over a wide range of SNR. The results indicate that the gap is less than $0.1\%$ for the SNR regime from $0$ to $30$ dB.

In the case of four stations, the exhaustive search for solving the maxmin problem becomes computationally expensive. Therefore, we use the upper bound of the maxmin problem calculated by the SDR scheme for benchmarking instead. Generally, the throughput gap between the equal-SNR and the upper bound is small, and becomes even smaller in the high SNR regime in Fig.~\ref{4u} and Table II. Since we use the upper bound for the maxmin solution for benchmarking, we conclude that the throughput gap between the maxmin problem and the equal-SNR problem is also small.

Therefore, the maxmin problem can be well approximated by the equal-SNR problem. We note that the trend as indicated by the simulations results in Fig.~\ref{2u}, Fig.~\ref{4u}, Table I and Table II is consistent with the analytical result of Proposition \ref{lem_nogap}. In the high SNR regime, the noise becomes negligible, and the gap between equal-SNR and maxmin diminishes.

\begin{figure}
\centering
\includegraphics[width=3.5in]{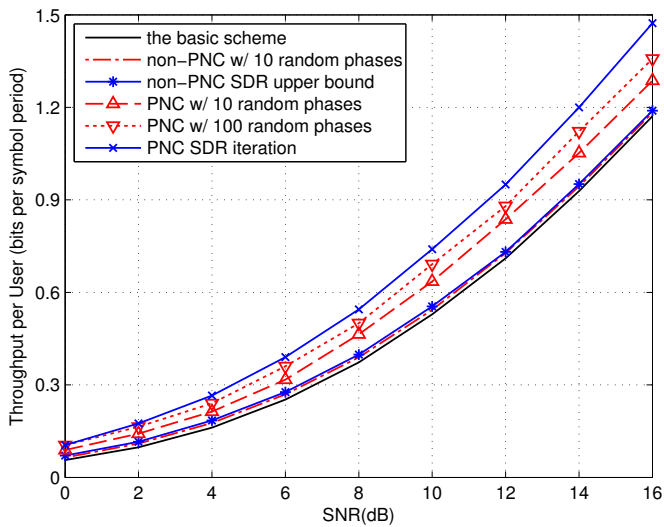}
\caption{Throughput comparison of different relaying schemes for non-pairwise switching pattern in the case of four stations.}
\label{4u_np}
\end{figure}

\begin{table*}
\caption{Worst-station Throughputs of the maxmin problem  (exhaustive search \& SDR) and throughputs of the equal-SNR problem (non-PNC opposite-phase \& PNC phase-aligned algorithms) when $N=2$; $\Delta$ denotes the variational ratio w.r.t. the equal-SNR solution of either non-PNC opposite-phase or PNC phase-aligned algorithms, respectively.} \centering
\begin{tabular}{c  c  c  c c c  c}
\cmidrule[1pt]{1-7}
 & \multicolumn{3}{c }{Non-PNC} &\phantom{abc}& \multicolumn{2}{c}{PNC}\\
\cmidrule[1pt]{2-4} \cmidrule[1pt]{6-7}
SNR & Equal-SNR & \multicolumn{2}{c }{Maxmin} &\phantom{abc}& Equal-SNR & \multicolumn{1}{ c}{Maxmin}\\
\cmidrule[1pt]{2-4} \cmidrule[1pt]{6-7}
(dB) & Opposite-phase & Optimal/$\Delta(\%)$ & SDR/$\Delta(\%)$  &\phantom{abc}& Phase-aligned & SDR/$\Delta(\%)$\\
\cmidrule[1pt]{1-7}
0      &   0.1270  &    0.1271 / 0.08  &    0.1270     / -0.00   &&  0.2305   &   0.2363 / 2.52\\
10     &   0.7990  &    0.7994 / 0.05  &    0.7987     / -0.03   &&  1.2091   &   1.2264 / 1.43\\
20     &   2.1249  &    2.1256 / 0.03  &    2.1202     / -0.22   &&  2.7250   &   2.7549 / 1.10\\
30     &   3.7075  &    3.7083 / 0.02  &    3.6864    /  -0.57   &&  4.3762   &   4.4016 / 0.58\\
\cmidrule[1pt]{1-7}
\end{tabular}
\end{table*}

\begin{table*}
\caption{Worst-station Throughputs of the maxmin problem  (SDR upper bound \& SDR) and throughputs of the equal-SNR problem (non-PNC opposite-phase \& PNC phase-aligned algorithms)  when $N=4$; $\Delta$ denotes the variational ratio w.r.t. the equal-SNR solution of either non-PNC opposite-phase or PNC phase-aligned algorithms, respectively.} \centering
\begin{tabular}{c  c  c  c c c  c}
\cmidrule[1pt]{1-7}
 & \multicolumn{3}{c }{Non-PNC} &\phantom{abc}& \multicolumn{2}{c}{PNC}\\
\cmidrule[1pt]{2-4} \cmidrule[1pt]{6-7}
SNR & Equal-SNR & \multicolumn{2}{c }{Maxmin} &\phantom{abc}& Equal-SNR & \multicolumn{1}{ c}{Maxmin}\\
\cmidrule[1pt]{2-4} \cmidrule[1pt]{6-7}
(dB) & Opposite-phase & SDR-Upper/$\Delta(\%)$ & SDR/$\Delta(\%)$  &\phantom{abc}&  Phase-aligned & SDR/$\Delta(\%)$\\
\cmidrule[1pt]{1-7}
0      &   0.0661  &    0.0720 / 8.85  &   0.0720     / 8.85   &&   0.1266   &   0.1376 / 8.69\\
10     &   0.5504  &    0.5562 / 1.05  &   0.5552     / 0.87   &&   0.8842   &   0.9018 / 1.99\\
20     &   1.7314  &    1.7328 / 0.08  &   1.7302     / -0.08  &&   2.3302   &   2.3435 / 0.57\\
30     &   3.2906  &    3.2918 / 0.04  &   3.2818    /  -0.27  &&   3.9611   &   3.9734 / 0.31\\
\cmidrule[1pt]{1-7}
\end{tabular}
\end{table*}

\medskip\noindent{\bf\emph{Observation 5}}: The throughput of the equal-SNR problem is roughly the same as the approximate throughput of the maxmin problem achieved by the SDR technique for pairwise switching pattern. However, for the non-pairwise pattern, the SDR technique achieves good throughput performance for both non-PNC and PNC relaying schemes.

With reference to Table I and Table II for pairwise switching pattern, for zero-forcing relaying our non-PNC opposite-phase algorithm is better than the SDR scheme in the high SNR regime, and for network-coded relaying the SDR scheme is better than our PNC phase-aligned algorithm. However, the gap is mostly smaller than $2\%$ for SNR larger than $0$ dB.

For the non-pairwise switching pattern, the random-phase scheme is close to the SDR upper bound for the non-PNC relaying. However, for the PNC relaying, the iterative SDR scheme outperforms the PNC identical-$b$ random-phase algorithm since the latter needs more trials to achieve better throughput performance.

The overall implications of our analytical and simulation results are as follows. The equal-SNR scheme, with the target of achieving perfect fairness among the links, is also a good approximation to the maxmin problem when the relay noise is small. Given a symmetric switch matrix that realizes pairwise transmissions, we could use the non-PNC opposite-phase algorithm or the non-PNC random-phase  algorithm for zero-forcing relaying, and the PNC phase-aligned  algorithm or the PNC identical-$b$ random-phase algorithm for network-coded relaying, to identify a suitable gain vector. The PNC phase-aligned algorithm has good throughput performance as well as fast execution time. Given an asymmetric switch matrix that realizes non-pairwise transmissions, we could use the random-phase algorithm for zero-forcing relaying and network-coded relaying to identify a suitable gain vector. The SDR scheme also achieves good throughput performance; however, its complexity is generally higher than the equal-SNR schemes.

\section{Conclusion}

We have proposed a framework for wireless MIMO switching to facilitate communication among multiple wireless stations. With optimized precoders, network-coded relaying improves the throughput performance significantly over non-network-coded relaying.

The maxmin solution and the equal-SNR solution have their respective advantages.
The former yields better throughput performance. However, the equal-SNR solution guarantees perfect fairness. With the equal-SNR solution, MIMO switching can be easily extended to multiple transmissions, by which general transmission patterns can be realized, including unicast, multicast, broadcast, or a mixture of them \cite{wms11}.
Moreover, the maxmin problem is NP-hard, which is solved using exhaustive search. Even for the approximation of the SDR scheme, SDP problems using interior point methods has the complexity cost at most $O(N^7)$ \cite{luo06}. However, the proposed schemes for the equal-SNR problem only has the complexity cost of $O(N^3)$, which is mainly induced by calculating matrix inverses. Hence, the equal-SNR problem is practical for implementation. For the above reasons, the equal-SNR setting is perhaps more amenable to practical deployment.

In future work, it will be interesting to explore switch matrices that realize more complicated patterns than unicast. It would also be interesting to study the case where the number of antennas at the relay is fewer than the number of stations.

\bibliographystyle{IEEEtran}
\bibliography{MIMO_switch}



\begin{IEEEbiography}[{\includegraphics[width=1in,height=1.25in,clip,keepaspectratio]{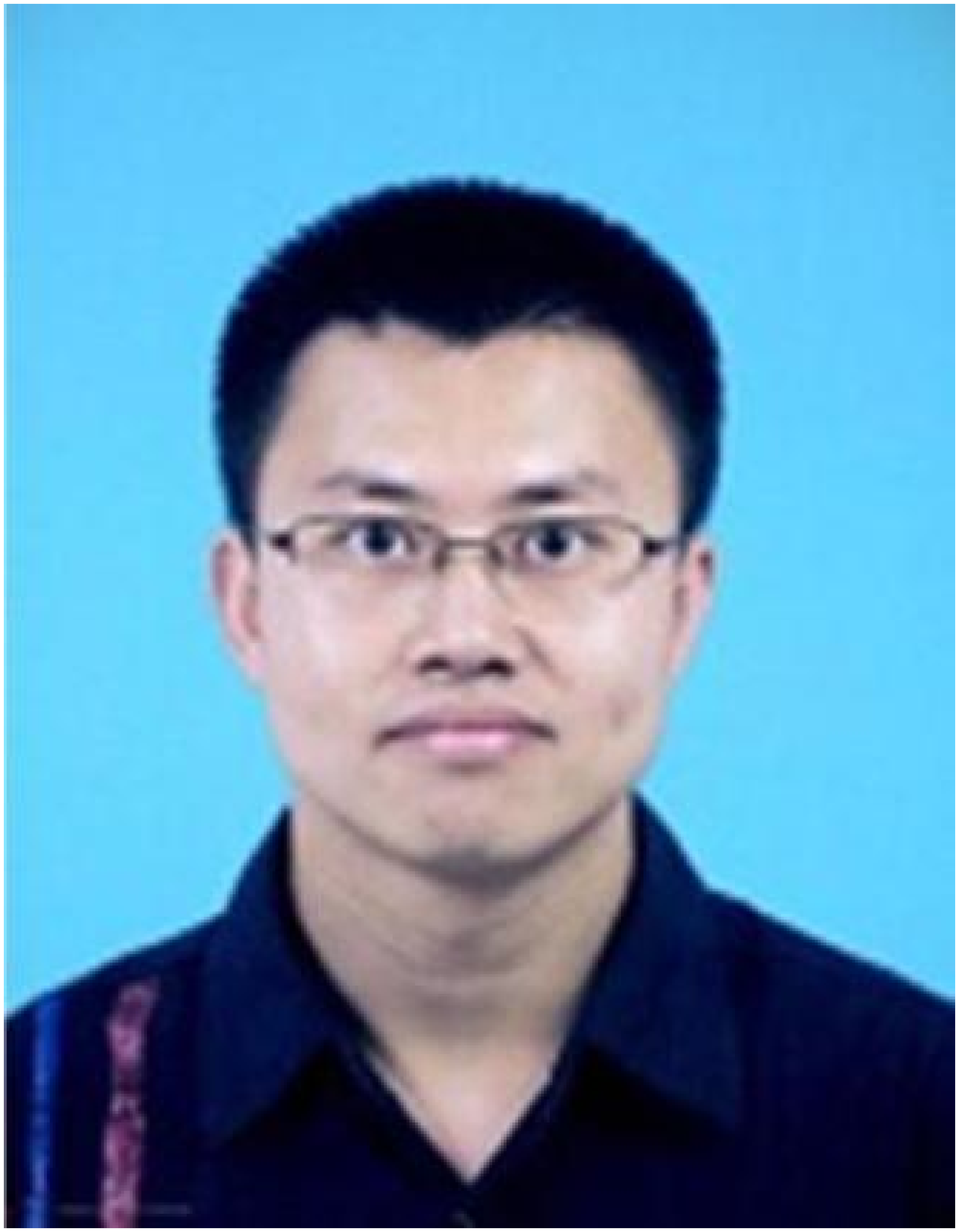}}]
{Fanggang Wang}
  (S'10-M'11) received the B.S. degree in 2005 and the Ph.D. degree in 2010 in the School of Information and Communication Engineering from Beijing University of Posts and Telecommunications, Beijing, China. From 2008 to 2010, he worked as a visiting scholar in Electrical Engineering Department, Columbia University, New York City, New York, USA. Since 2010, he has been working in the State Key Lab of Rail Traffic Control and Safety, School of Electronic and Information Engineering, Beijing Jiaotong University as an Assistant Professor, and also in Institute of Network Coding, The Chinese University of Hong Kong as a Postdoctoral Fellow. In 2011, he visited Department of Electrical Engineering \& Computer Science, Northwestern University, Evanston, IL, USA. His research interests are in the area of MIMO, OFDM and network coding techniques in wireless communications.
  
  He chaired two workshops on wireless network coding (NRN 2011 \& NRN 2012) and served as TPC member in several conferences.
\end{IEEEbiography}

\begin{IEEEbiography}[{\includegraphics[width=1in,height=1.25in,clip,keepaspectratio]{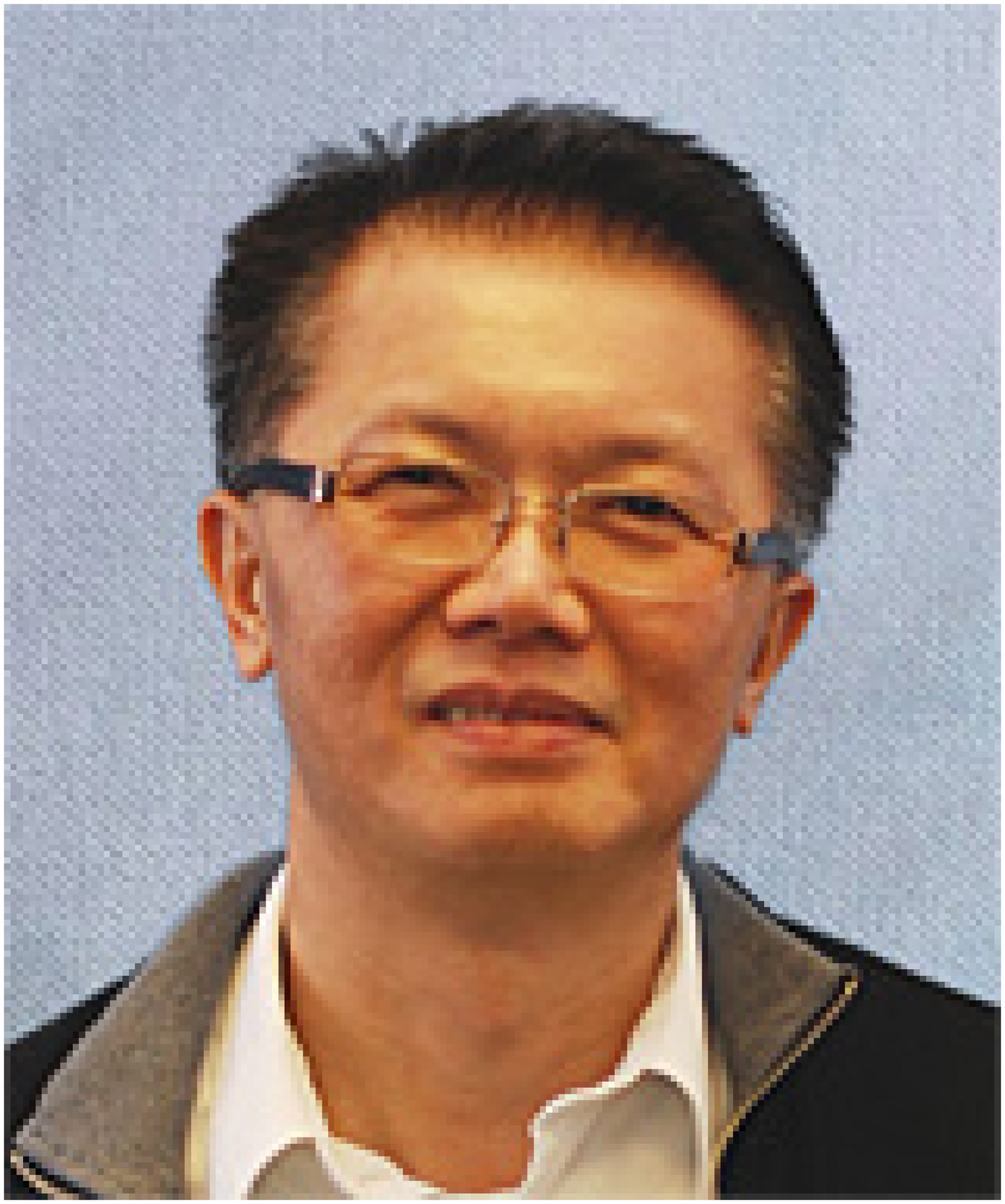}}]
{Soung Chang Liew}
received his S.B., S.M., E.E., and Ph.D. degrees from the Massachusetts Institute of Technology. From 1984 to 1988, he was at the MIT Laboratory for Information and Decision Systems, where he investigated Fiber-Optic Communications Networks. From March 1988 to July 1993, he was at Bellcore (now Telcordia), New Jersey, where he engaged in Broadband Network Research. He has been a Professor at the Department of Information Engineering, the Chinese University of Hong Kong, since 1993.  He is an Adjunct Professor at Peking University and Southeast University, China.

Prof. Liew's current research interests include wireless networks, Internet protocols, multimedia communications, and packet switch design. Prof. Liew's research group won the best paper awards in IEEE MASS 2004 and IEEE WLN 2004. Separately, TCP Veno, a version of TCP to improve its performance over wireless networks proposed by Prof. Liew's research group, has been incorporated into a recent release of Linux OS. In addition, Prof. Liew initiated and built the first inter-university ATM network testbed in Hong Kong in 1993. More recently, Prof. Liew's research group pioneers the concept of Physical-layer Network Coding (PNC).

Besides academic activities, Prof. Liew is also active in the industry. He co-founded two technology start-ups in Internet Software and has been serving as a consultant to many companies and industrial organizations. He is currently consultant for the Hong Kong Applied Science and Technology Research Institute (ASTRI), providing technical advice as well as helping to formulate R\&D directions and strategies in the areas of Wireless Internetworking, Applications, and Services.

Prof. Liew is the holder of eight U.S. patents and a Fellow of IEEE, IET and HKIE. He currently serves as Editor for IEEE Transactions on Wireless Communications and Ad Hoc and Sensor Wireless Networks. He is the recipient of the first Vice-Chancellor Exemplary Teaching Award at the Chinese University of Hong Kong. Publications of Prof. Liew can be found in www.ie.cuhk.edu.hk/soung.
\end{IEEEbiography}

\begin{IEEEbiography}[{\includegraphics[width=1in,height=1.25in,clip,keepaspectratio]{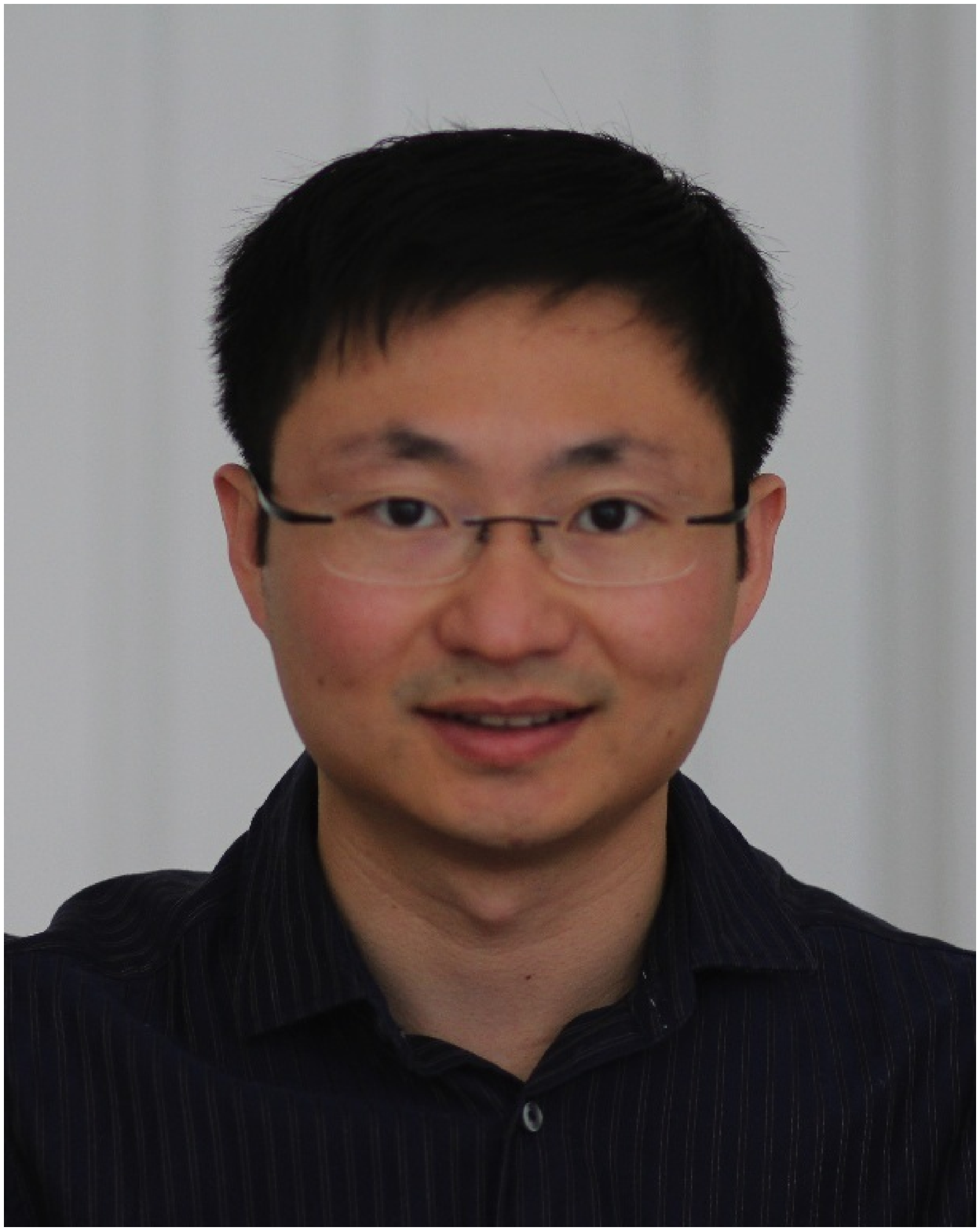}}]{Dongning Guo}
(S'97-M'05-SM'11) joined the faculty of Northwestern
University, Evanston,
IL, in 2004, where he is currently an Associate Professor in the
Department of Electrical Engineering and Computer Science. He
received the B.Eng.\ degree from the University of Science \&
Technology of China, the M.Eng.\ degree from the National University
of Singapore, and the M.A.\ and Ph.D.\ degrees from Princeton
University, Princeton, NJ. He was an R\&D Engineer in the Center for
Wireless Communications (now the Institute for Infocom Research),
Singapore, from 1998 to 1999. He has held visiting positions at
Norwegian University of Science and Technology in summer 2006 and in
the Institute of Network Coding at the Chinese University of Hong
Kong in 2010--2011. He is an Associate Editor of the IEEE Transactions on
Information Theory and an Editor of Foundations and Trends in
Communications and Information Theory.

Dongning Guo received the Huber and Suhner Best Student Paper Award in
the International Zurich Seminar on Broadband Communications in 2000
and is a co-recipient of the IEEE Marconi Prize Paper Award in
Wireless Communications in 2010 (with Y.~Zhu and M.~L.~Honig). He is
also a recipient of the
National Science Foundation Faculty Early Career Development (CAREER)
Award in 2007. His research interests are in information theory,
communications, and networking.
\end{IEEEbiography}

\end{document}